\RequirePackage{fix-cm}
\documentclass{svjour3}                     
\smartqed  
\usepackage{graphicx}
\usepackage{latexsym}
\usepackage{amssymb}
\usepackage{amsmath}
\usepackage{enumerate}
\usepackage{geometry}
\usepackage[colorlinks=true,citecolor=black,linkcolor=black,urlcolor=blue]{hyperref}
\usepackage[justification=centering]{caption}
\usepackage{amsmath,bm}

\newtheorem{fact}{Fact}
\journalname{Designs, Codes and Cryptography}
\begin{document}

\title{New Characterizations for the Multi-output Correlation-Immune Boolean Functions \thanks{This research is supported in part by NSFC (No. 61671013, 61672410, 61602361), NSF of Shaanxi Province (No. 2018JM6076), and the Programme of Introducing Talents of Discipline to Universities (China 111 Project, No. B16037)}
}


\author{Jinjin Chai         \and
        Zilong Wang         \and
        Sihem Mesnager      \and
        Guang Gong
}


\institute{Jinjin Chai and Zilong Wang\at
              State Key Laboratory of Integrated Service Networks\\
              School of Cyber Engineering, Xidian University \\
              \email{jj\_chai@163.com,\ zlwang@xidian.edu.cn}           
           \and
           Sihem Mesnager \at
              Department of Mathematics\\
              University of Paris VIII, University of Paris XIII\\
              \email{smesnager@univ-paris8.fr}
           \and
           Guang Gong \at
              Department of Electrical and Computer Engineering\\
              University of Waterloo\\
              \email{ggong@uwaterloo.ca}\\
}

\date{Received: date / Accepted: date}

\maketitle

\begin{abstract}
Correlation-immune (CI) multi-output Boolean functions have the property of keeping the same output distribution when some input variables are fixed. Recently, a new application of CI functions has appeared in the system of resisting side-channel attacks (SCA). In this paper, three new methods are proposed to characterize the $t$ th-order CI multi-output Boolean functions ($n$-input and $m$-output). The first characterization is to regard the multi-output Boolean functions as the corresponding generalized Boolean functions. It is shown that a
generalized Boolean functions $f_g$ is a $t$ th-order CI function if and only if the Walsh transform of $f_g$ defined here vanishes at  all points with Hamming weights between $1$ and $t$. Compared to the previous Walsh transforms of component functions, our first method can reduce the computational complexity from $(2^m-1)\sum^t_{j=1}\binom{n}{j}$ to $m\sum^t_{j=1}\binom{n}{j}$. The last two methods are generalized from Fourier spectral characterizations. Especially, Fourier spectral characterizations are more efficient to characterize the symmetric multi-output CI Boolean functions.


\end{abstract}

\section{Introduction}\label{section:s1}
The correlation-immune (CI) functions were originally used to resist Siegenthaler's correlation attack~\cite{siegenthaler1985}(or {\lq divide and conquer attack \rq}~\cite{Siegenthaler1984}) in stream ciphers in the last century. The correlation immunity of functions gradually loses its interest with the development of new attacks. But, recently, in paper~\cite{carlet2014Correlation-immune,Carlet2018Constructing}, a new application of CI functions has appeared in the system of resisting side-channel attacks (SCA), which has renewed interest. These attacks on the implementations of block ciphers in embedded systems like smart cards, FPGA or ASIC assume an attacker model different from classical attacks, and are extremely powerful in practice. These implementations then need to include counter-measures, which reduces the efficiency of the cryptosystem and adds additional storage. The CI functions allow cost reduction of counter-measures to SCA. Moreover, these functions need to have low Hamming weights. 


We focus on the characterization of CI functions. 
A characterization of CI Boolean functions was obtained by Xiao and Massey~\cite{Xiao1988A} in terms of the Walsh transform in 1988. That is, a Boolean function is $t$ th-order CI if and only if its Walsh transform vanishes for all points with Hamming weights between $1$ and $t$. In 1959 Golomb~\cite{Golomb1959On} introduced the concept of the invariants of Boolean functions in order to classify Boolean functions. This work was collected in his book {\em Shift Register Sequences~\cite{Golomb1967Shift}}, Chapter VIII. Golomb did not mention the original applications for cryptography of his work on invariants until his paper~\cite{Golomb1999On} published in 1999. In fact, his work is the same concept with the Walsh spectral characterization of CI Boolean functions. It has been proposed in~\cite{CarletBoolean} to call this the Golomb-Xiao-Massey characterization. The Golomb-Xiao-Massey characterization of multi-output correlation immune functions comes directly from the one of correlation immune Boolean functions. That is, a multi-output Boolean function is $t$ th-order CI if and only if all its nonzero linear combinations of the component functions are $t$ th-order CI.
In addition to Golomb-Xiao-Massey characterization, other methods to characterize CI functions, such as matrices~\cite{Gopalakrishnan1995,camion1999correlation}, orthogonal arrays~\cite{Camion1991On,Bierbrauer96}, and the Fourier spectra~\cite{WangDiscrete2018,wang2019TheFourier} were also proposed.


Since there is a natural one-to-one correspondence between vectors in $\mathbb{F}_2^m$ and integers in $[0,2^m-1]$, we can represent a multi-output Boolean function as a corresponding generalized Boolean function. Schmidt \cite{Schmidt2006Quaternary} gave the $2$-adic expansion for a generalized Boolean function (this expansion is unique), and used it to study generalized bent functions that are applied in MC-CDMA. Similarly, we use this representation to get new characterizations for multi-output CI Boolean
functions. Our first characterization shows that a multi-output Boolean function is a $t$ th-order CI Boolean function if and only if the Walsh transform of the corresponding generalized Boolean function $f_g$ defined in this paper vanishes at all points with Hamming weights between $1$ and $t$. Compared to the previous Walsh spectral characterization method, this method reduces the complexity of calculations from $(2^m-1)\sum^t_{j=1}\binom{n}{j}$  to $m\sum^t_{j=1}\binom{n}{j}$ to determine whether a function is $t$ th-order CI. Wang and Gong ~\cite{WangDiscrete2018} investigated discrete Fourier transform  of (single-output) Boolean functions and deduced an equivalent condition for  $t$ th-order CI Boolean functions. Fourier spectral characterizations are generalized here to characterize the $t$ th-order multi-output CI Boolean functions. And these Fourier spectral characterizations are much more efficient to characterize the symmetric multi-output CI Boolean functions.


The rest of the paper is organized as follows. In Section~\ref{section:s3}, we introduce three representations of multi-output Boolean functions, the definitions of the correlation immunity, as well as the Walsh transform of the multi-output Boolean functions.
In Section~\ref{section:s4}, we present three new characterizations for multi-output CI Boolean functions.
The first characterization is in terms of the Walsh transform and the last two characterizations are in terms of the Fourier transforms over the complex field. Section~\ref{section:s7} concludes the paper.

\section{Preliminaries}\label{section:s3}
The following notations will be used throughout the paper.
\begin{description}
  \item[-] $n$ and $m$ are positive integers.
  \item[-] $\mathbb{F}_{2^m}$ is a finite field with $2^m$ elements. $\mathbb{Z}_{2^m}$ is a residue class ring of integers modulo $2^m$.
  \item[-] $\mathbb{F}_2^{m*}=\mathbb{F}_2^m \setminus \{\mathbf{0}\}$.
  \item[-] For $\bm{c}=(c_1,c_2,\cdots,c_n)\in \mathbb{F}_2^n$, $wt(c)$ denotes the {\em Hamming weight} of $\bm{c}$, i.e., the number of nonzero terms in $\bm{c}$.
  \item[-] For $1\leq i \leq m$, $\omega_i=exp(\frac{2\pi\sqrt{-1}}{2^i})$ is a $2^ith$ primitive root of unity over the complex field.
  \item[-] $\#\{\cdot\}$ denotes the number of elements in the set $\{\cdot\}$.
  \item[-] $S_n$ is a symmetric group consisting of permutations of the set $\{1,2,\cdots ,n\}$. $\pi\in S_n$ is a permutation of symbols $\{1,2,\cdots ,n\}$.
\end{description}

\subsection{The Representations of Multi-output Boolean Functions}
Here we give three representations of a multi-output Boolean function: representations as component functions, as a trace function, and as a generalized Boolean function. We shall introduce them respectively in this section.

A {\em Boolean function} is a function $f_b$: $\mathbb{F}^n_2\rightarrow \mathbb{F}_2$ with variable $\bm{x}=(x_1,x_2,\cdots,x_n)$, where $\mathbb{F}_2$ is the finite field with two elements, and $\mathbb{F}_2^n$ is $n$-dimension vector space over $\mathbb{F}_2$. It can be represented by its {\em algebraic normal form} (ANF):
\[f_b(\bm{x})=\sum^{2^n-1}_{k=0}c_k\prod^{n}_{j=1}x_j^{k_j},c_k\in\mathbb{F}_2,\]
where $(k_1,k_2,\cdots,k_n)$ is the binary expansion of $k$.

A $n$-input and $m$-output {\em multi-output Boolean function} can be represented as a function from $\mathbb{F}^n_2$ to $\mathbb{F}^m_2: f(\bm{x})=(f_1(\bm{x}),f_2(\bm{x}),\cdots,f_m(\bm{x}))$. Every component function $f_i$, $1 \leq i \leq m$, is a Boolean function. Obviously, $f(\bm{x})$ is a (single-output) Boolean function when $m = 1$. Multi-output Boolean functions are also called {\em vectorial Boolean functions}. 
We will refer to a multi-output Boolean function as an $(n,m)$-function for simplicity.

A multi-output Boolean function can be represented as a trace function from $\mathbb{F}_{2^n}$ to $\mathbb{F}_{2^m}$ when $m$ is a divisor of $n$. The vector space $\mathbb{F}_2^n$ can be endowed with the structure of the field $\mathbb{F}_{2^n}$. Then any multi-output Boolean function $f(\bm{x})$ can be viewed as a function from $\mathbb{F}_{2^n}$ to $\mathbb{F}_{2^m}$ ($\mathbb{F}_{2^m}$ is a sub-field of $\mathbb{F}_{2^n}$).
A multi-output Boolean function $f(\bm{x})$ can be represented in the form
$$\mbox{Tr}_m^n(\sum_{j=0}^{2^n-1}\delta_jx^j), \delta_j\in \mathbb{F}_{2^n},$$ where $\mbox{Tr}_m^n(x)=x+x^{2^m}+x^{2^{2m}}+x^{2^{3m}}+\cdots +x^{2^{n-m}}$ is the trace function from $\mathbb{F}_{2^n}$ to $\mathbb{F}_{2^m}$.

A {\em generalized Boolean function} $f_g$ is a function from $\mathbb{F}_2^n$ to $\mathbb{Z}_{2^m}$. Such a function can be uniquely expressed as a linear combination of the monomials \[1, x_1, x_2, \cdots, x_n, x_1x_2, x_1x_3, \cdots, x_{n-1}x_n, \cdots, x_1x_2x_3\cdots x_n,\]
where the coefficient of each monomial belongs to $\mathbb{Z}_{2^m}$. 
Since there is a natural one-to-one correspondence between the vectors in $\mathbb{F}_2^m$ and the elements in $\mathbb{Z}_{2^m}$, we can represent a multi-output Boolean function as its corresponding generalized Boolean function, i.e.,
\begin{equation}\label{gen-func}
f_g(x_1,x_2,\cdots,x_n)=\sum^{m}_{i=1}2^{i-1}f_i.
\end{equation}
Such representation was used to study generalized bent functions by Schmidt \cite{Schmidt2006Quaternary}.

Note that in formula (\ref{gen-func}), each $f_i$ is a Boolean function which is calculated modulo $2$, while the summation $\sum^{m}_{i=1}2^if_i$ is calculated modulo $2^m$. So the algebraic normal form of the general generalized Boolean function $f_g$ cannot be directly obtained from the weighted sum of the algebraic normal of the component Boolean function $f_i$. For example, let $f(x_1, x_2, x_3)=(f_1, f_2)$ be a multi-output Boolean function where
\begin{equation*}
\begin{aligned}
f_1 & =  x_1x_2, \\
f_2 & =  x_2+x_3.
\end{aligned}
\end{equation*}
It is clear $2f_2+f_1=x_1x_2+2x_2+2x_3$. However, from the truth table of the function $f(x_1, x_2, x_3)=(f_1, f_2)$, we have
$$f_g=2x_1x_2+2x_2x_3+x_2+x_3.$$

\begin{table}[htbp]

\begin{center}
\caption{Truth table of the function $f(x_1, x_2, x_3)=(f_1, f_2)$}\label{table:t1}
\begin{tabular}{ccc|c|c}
  \hline
  $x_{3}$ & $x_{2}$ & $x_{1}$ & $(f_1, f_2)$ & $f_g(x_1, x_2, x_3)$ \\
  \hline 0 & 0 & 0 & 00 & 0 \\
  \hline 0 & 0 & 1 & 00 & 0 \\
  \hline 0 & 1 & 0 & 01 & 1 \\
  \hline 0 & 1 & 1 & 11 & 3 \\
  \hline 1 & 0 & 0 & 01 & 1 \\
  \hline 1 & 0 & 1 & 01 & 1 \\
  \hline 1 & 1 & 0 & 00 & 0 \\
  \hline 1 & 1 & 1 & 10 & 2 \\
  \hline
\end{tabular}
\end{center}
\end{table}
For more information about correlation-immune Boolean and vectorial functions, we invite the readers to consult the excellent chapters provided by Carlet~\cite{Carlet2010Boolean,Carlet2010vectorial}.

\subsection{Correlation Immunity}

The multi-output CI Boolean functions are defined initially from the perspective of probability theory, which is similar to the definition of CI of single-output Boolean functions.

\begin{definition}\label{definition:d1}
Let $t$ be an integer such that $0 \leq t \leq n$. An $(n, m)$-function $f(\bm{x})$ is called $t$ th-order CI if its output distribution does not change when at most $t$ coordinates $x_i$ of $\bm{x}$ are kept constant. In other words,
\begin{flalign}\label{Pr}
\begin{split}
   P_r(f(x_1,x_2,\cdots,x_n) &=  (y_1,y_2,\cdots,y_m)|x_{i_j}=a_j,1\leq j\leq t) \\
   & =P_r(f(x_1,x_2,\cdots,x_n)=(y_1,y_2,\cdots,y_m))
\end{split}
\end{flalign}
for every $t$-subset $\{i_1,\cdots ,i_t\}\subseteq \{1,\cdots ,n\}$, $a_j\in\mathbb{F}_2(1\leq j\leq t)$, and $(y_1,y_2,\cdots,y_m)\in \mathbb{F}_2^m$.
\end{definition}

We will refer to a $t$ th-order CI $(n, m)$-function as $(n, m, t)$-CI function for simplicity. 
The {\em Walsh transform} of an $(n, m)$-function $f(\bm{x})$ is the function which maps any ordered pair $(\bm{u}, \bm{v}) \in \mathbb{F}_2^n \times \mathbb{F}_2^{m*}$ to the value at $\bm{u}$ of the Walsh transform of the component function $\bm{v} \cdot f(\bm{x})$,  that is,
\begin{equation}\label{walshtrans}
\hat{f}(\bm{u},\bm{v})=\sum_{\bm{x} \in \mathbb{F}_2^n}(-1)^{\bm{v} \cdot f(\bm{x}) +\bm{u}\cdot \bm{x} }.
\end{equation}

\begin{fact}\label{fact1}
An $(n, m)$-function is an $(n, m, t)$-CI function if and only if $\hat{f}(\bm{u},\bm{v})=0$ for $\bm{v}\neq \bm{0}$, $1 \leq wt(\bm{u}) \leq t$, where $wt(\bm{u})$ denotes the Hamming weight of $\bm{u}$.
\end{fact}

If we consider an $(n, m)$-function by a generalized Boolean Function, then equation~(\ref{Pr}) shall be rewritten as
\begin{flalign}
\begin{split}
   P_r(f_g(x_1,x_2,\cdots,x_n) &=  \alpha |x_{i_j}=a_j,1\leq j\leq t) \\
   & =P_r(f_g(x_1,x_2,\cdots,x_n)=\alpha)
\end{split}
\end{flalign}
for every $t$-subset $\{i_1,\cdots ,i_t\}\subseteq \{1,\cdots ,n\}$, $a_j\in\mathbb{F}_2(1\leq j\leq t)$, and $\alpha\in \mathbb{Z}_{2^m}$.

\section{New Characterizations}\label{section:s4}

In this section, we present three new characterizations for multi-output CI Boolean functions. Our first two characterizations shall consider the multi-output Boolean functions as generalized Boolean functions. The last characterization considers the multi-output Boolean functions from the perspective of component functions.

\subsection{The First Characterization}

We give a new method to characterize a multi-output CI Boolean function in terms of its corresponding generalized Boolean function's Walsh transform.

\begin{theorem}\label{theorem:t1}
Let $f_g(x_1,x_2,\cdots ,x_n)$ be a generalized Boolean function. Then $f_g(x_1,x_2,\cdots ,x_n)$ is an $(n, m, t)$-CI function if and only if
\begin{equation}\label{equation:e1}
  \sum_{\bm{x}\in \mathbb{F}_2^n}\omega_i^{f_g(\bm{x})}(-1)^{\bm{c}\cdot\bm{x}}=0,
\end{equation}
for $1\leq wt(\bm{c})\leq t$ and $1 \leq i \leq m$, where $\omega_i$ is a $2^ith$ primitive root of unity in the complex field.
\end{theorem}
We  introduce the {\lq linear combination lemma\rq} \cite{Xiao1988A,Brynielsson1989A} before proving theorem \ref{theorem:t1}.

\begin{fact}\label{lemma:l2}
\cite{Xiao1988A} The discrete random variable $Z$ is independent of the $k$ independent binary random variables $\mathbf{X}=(X_1,X_2,\cdots,X_k)$ if and only if $Z$ is independent of the sum $c_1X_1+c_2X_2+,\cdots,+c_kX_k$ for every choice of $\mathbf{c}=(c_1,c_2,\cdots,c_k)\in\mathbb{F}_2^{k*}$. 
\end{fact}
%
Now, we shall prove Theorem \ref{theorem:t1} by using Fact \ref{lemma:l2}.
\begin{proof}
The equation (\ref{equation:e1}) can be divided into two parts. One is for $\bm{c}\cdot\bm{x}=0$, and the other is for $\bm{c}\cdot\bm{x}=1$, i.e.,
\begin{equation}\label{equation:e2}
  \sum_{\bm{c}\cdot\bm{x}=0}\omega_i^{f_g(\bm{x})}-\sum_{\bm{c}\cdot\bm{x}=1}\omega_i^{f_g(\bm{x})}=0.
\end{equation}
We denote that
\[a_{\alpha} =\# \{\bm{x}:f_g(\bm{x})=\alpha,\bm{c}\cdot\bm{x}=0\},\]
and
\[b_{\alpha} =\# \{\bm{x}:f_g(\bm{x})=\alpha,\bm{c}\cdot\bm{x}=1\},\]
where $0 \leq \alpha \leq 2^m-1$.
Hence, equation (\ref{equation:e2}) is equivalent to
\begin{equation*}
  \sum_{\alpha=0}^{2^m-1}a_{\alpha}\omega_i^{\alpha}-\sum_{\alpha=0}^{2^m-1}b_{\alpha}\omega_i^{\alpha}=0 \Longleftrightarrow \sum_{\alpha=0}^{2^m-1}(a_{\alpha}-b_{\alpha})\omega_i^{\alpha}=0.
\end{equation*}
For any integer $d$, let $\Phi_d(z)$ be the  $d$th {\em cyclotomic polynomial} \cite{mceliece1987finite}, which is a monic polynomial of degree $\phi(d)$ (Euler function). It is known that
\[\Phi_{2^i}(z)=\prod \{(z-\xi^{j}):0\leq j\leq 2^n-1,\mbox{gcd}(j,2^n)=2^{n-i}\},\]
where gcd denotes the great common divisor and $\xi=exp(\frac{2\pi \sqrt{-1}}{2^n})$. We have
\[\Phi_2(z)=z+1,\Phi_4(z)=z^2+1,\cdots,\Phi_{2^m}(z)=z^{2^{m-1}}+1.\]
For $1 \leq i \leq m$, $\Phi_{2^i}(z)$ is a monic polynomial with integer coefficients that is the minimal polynomial over the rational field of any primitive $2^i$th root of unity. Since $\omega_i$ is a $2^ith$ primitive root of unity in the complex field, and $\Phi_{2^i}(z)$ is irreducible in the integer ring, then every $\Phi_{2^i}(z)$ divide $\sum_{\alpha=0}^{2^m-1}(a_{\alpha}-b_{\alpha})z^{\alpha}$. In addition, $\Phi_{2^i}(z)$ are pairwise coprime for $1 \leq i \leq m$.
 Let $h(z)$ denote the product of $\Phi_{2^i}(z)$ for $1 \leq i \leq m$, i.e.,
$$h(z)=(z+1)(z^2+1)\cdots(z^{2^{m-1}}+1)=\frac{z^{2^m}-1}{z-1}=1+z+z^2+\cdots+z^{2^m-1}.$$
Note that $\sum_{\alpha=0}^{2^m-1}(a_{\alpha}-b_{\alpha})z^{\alpha}$ must be a multiple of $h(z)$, and $\mbox{deg}(\sum_{\alpha=0}^{2^m-1}(a_{\alpha}-b_{\alpha})z^{\alpha})=\mbox{deg}(h(z))$, we obtain that $$a_1-b_1=a_2-b_2=\cdots=a_{2^m-1}-b_{2^m-1}.$$
Since \[\sum_{\alpha=0}^{2^m-1}a_\alpha=\sum_{\alpha=0}^{2^m-1}b_\alpha=2^{n-1}\Rightarrow \sum_{\alpha=0}^{2^m-1}(a_\alpha-b_\alpha)=0,\]
we obtain that
$$a_\alpha-b_\alpha=0\Rightarrow a_\alpha=b_\alpha.$$
Thus, $f_g(\bm{x})$ is independent of $\bm{c}\cdot\bm{x}$ for $1\leq wt(\bm{c})\leq t$. Then we get $f_g(\bm{x})$ is independent of $x_{i_1},x_{i_2},\cdots,x_{i_t}$ according to Fact \ref{lemma:l2}. In other words,
\[P_r \left( f_g(\bm{x})=\alpha|x_{i_1}, x_{i_2} \cdots x_{i_t}\right) =P_r \left(f_g(\bm{x})=\alpha\right).\]
which is exactly the definition of the $(n, m, t)$-CI function.
\qed
\end{proof}

Compared to the previous Walsh spectral characterization (Fact \ref{fact1}), this characterization reduces the computational complexity from $(2^m-1)\sum^t_{j=1}\binom{n}{j}$  to $m\sum^t_{j=1}\binom{n}{j}$.

\subsection{The Second Characterization}\label{section:s5}


The second characterization is in terms of Fourier spectra of sequences described by the corresponding generalized Boolean functions. We first introduce the concept of the discrete Fourier transform (DFT) over the complex field of the sequences. Note that DFT over the complex field introduced here is the traditional DFT, which is different from the DFT over the finite field \cite{Golomb2005Signal}.

We describe a sequence $\bm{f}_g$ of length $2^n$ corresponding to a generalized Boolean function $f_g$ by listing the values taken by $f_g(x_1,x_2,\cdots ,x_n)$ as $(x_1,x_2,\cdots ,x_n)$ which ranges over all its $2^n$ values in lexicographic order. In other words, sequence $\bm{f}_g$ is defined by
$$\bm{f}_g=(f_g(0),f_g(1),\cdots ,f_g(2^n-1)),$$
where $f_g(k)=f_g(x_1,x_2,\cdots ,x_n)$ and $(x_1,x_2,\cdots ,x_n)$ is the binary representation of the integer $k$  for $0\leq k \leq 2^n-1$, i.e., $k=\sum_{i=1}^n x_i 2^{i-1}$. For example, for $n=3$ and $m=2$ we have $\bm{3x_1x_2x_3}=(000000023)$ and $\bm{2x_1x_2+3x_3+1}=(10101032)$ respectively.

Let $\omega_i$ be a $2^ith$ primitive root of unity over the complex field for $1\leq i \leq m$. The polynomials associated with sequences (every sequence defined by the generalized Boolean function $f_g$) are given by
\begin{equation}\label{Z-transform}
F^{\{i\}}(z)=\sum^{2^n-1}_{k=0}\omega_i^{f_g(k)}z^k, 1\leq i \leq m.
\end{equation}

\begin{definition}\label{def:d2}
Let $\xi=exp(\frac{2\pi\sqrt{-1}}{2^n})$ be a $2^n$th primitive root of unity over the complex field. The {\em discrete Fourier transform } (DFT) of sequences (every sequence defined by the generalized Boolean function $f_g$) over the complex field are defined by
\begin{equation}\label{DFT}
\mathcal{F}^{\{i\}}_{f_g}(j) =\sum^{N-1}_{k=0}\omega_i^{f_g(k)}\xi^{-kj}, 0\leq j \leq N-1,
\end{equation}
where $1\leq i \leq m$, $\omega_i$ is a $2^ith$ primitive root of unity over the complex field. 
\end{definition}
It is obvious that the equation (\ref{DFT}) is the DFT of a sequence defined by a Boolean function when $m=1$.
Let $\pi \cdot f_g=f_g(x_{\pi(1)},x_{\pi(2)},\cdots ,x_{\pi(n)})$ be a function obtained by permuting the variables in $f_g$, and $\pi \cdot F(z)$ be the polynomial associated with the function $\pi \cdot f_g$. Then the Fourier spectral characterization is given below.
\begin{theorem}\label{theorem:t2}
Let $f_g(x_1,x_2,\cdots ,x_n)$ be a generalized Boolean function. Then $f_g(x_1,x_2,\cdots ,x_n)$ is an $(n, m, t)$-CI function if and only if  $$\mathcal{F}^{\{i\}}_{\pi \cdot f_g}(2^{n-t})=0,$$ for $\forall \pi \in S_n$ and $1\leq i \leq m$.
\end{theorem}







\begin{proof}
Recall the polynomials $F^{\{i\}}(z)$ in equation (\ref{Z-transform}) and the definition of DFT in Definition \ref{def:d2}, we have
$$\mathcal{F}^{\{i\}}_{\pi \cdot f_g}(2^{n-t})=\pi \cdot F^{\{i\}}(\xi^{-2^{n-t}}).$$
Since the minimal polynomial of $\xi^{-2^{n-t}}$ over the rational field is $\Phi_{2^t}(z)$, we obtain that $\mathcal{F}^{\{i\}}_{\pi \cdot f_g}(2^{n-t})=\pi \cdot F^{\{i\}}(\xi^{-2^{n-t}})=0$ is equivalent to the fact that $\Phi_{2^t}(z)|(\pi \cdot F^{\{i\}}(z))$. We first consider permutation $\pi$ to be identity. Since
\begin{equation*}
F^{\{i\}}(z)=\sum^{2^n-1}_{k=0}\omega_i^{f_g(k)}z^k=\sum_{\bm{x}\in \mathbb{F}^n_2}\omega_i^{f_g(\bm{x})}\prod^n_{i=1}(z^{2^{i-1}})^{x_i}, 1 \leq i \leq m,
\end{equation*}
we have
$$ \Phi_{2^t}(z)|F^{\{i\}}(z)\Longleftrightarrow F^{\{i\}}(z)\equiv 0~(\mbox{mod}\ \Phi_{2^t}(z))\Longleftrightarrow \sum_{\bm{x}\in \mathbb{F}^n_2}\omega_i^{f_g(\bm{x})}\prod^n_{i=1}(z^{2^{i-1}})^{x_i}\equiv 0 ~(\mbox{mod}\  \Phi_{2^t}(z)).$$
From the definition of the cyclotomic polynomial, we know
$\Phi_{2^t}(z)=z^{2^{t-1}}+1$, so
$$\Phi_{2^t}(z)|z^{2^i}-1, \mbox{for} \ \forall i\geq t. $$
Then we have
\begin{equation}\label{middle}
\Phi_{2^t}(z)|F^{\{i\}}(z)\Longleftrightarrow \sum_{\bm{x}\in \mathbb{F}^n_2}\omega_i^{f_g(\bm{x})}\prod^t_{i=1}(z^{2^{i-1}})^{x_i}\equiv 0 ~(\mbox{mod}\  z^{2^{t-1}}+1).
\end{equation}
Then the summation in (\ref{middle}) can be divided into two parts, where the first part is for $x_t=0$ and the second part is for $x_t=1$. Hence $\Phi_{2^t}(z)|F^{\{i\}}(z)$, $1\leq i \leq m$, is equivalent to
\[\sum_{x_1,...,x_{t-1},x_t=0}\omega_i^{f_g(\bm{x})}\prod^{t-1}_{i=1}(z^{2^{i-1}})^{x_i}
 -\sum_{x_1,...,x_{t-1},x_t=1}\omega_i^{f_g(\bm{x})}\prod^{t-1}_{i=1}(z^{2^{i-1}})^{x_i}=0.\]
Combining like terms about $z$, the above condition is equivalent to

\[\sum_{x_1,...,x_{t-1}} \left(\sum_{x_t=0,x_{t+1},...,x_n}\omega_i^{f_g(\bm{x})}
                                          - \sum_{x_t=1,x_{t+1},...,x_n}\omega_i^{f_g(\bm{x})} \right) (z^{2^{i-1}})^{x_i}=0, \]
so the coefficients of $(z^{2^{i-1}})^{x_i}$ are
\begin{equation}\label{eq:wi}
  \sum_{x_t=0,x_{t+1},\cdots ,x_n} \omega_i^{f_g(\bm{x})}- \sum_{x_t=1,x_{t+1},\cdots,x_n}\omega_i^{f_g(\bm{x})}=0.
\end{equation}
Now, we denote that
\[a_{\alpha} =\# \{\bm{x}:f_g(\bm{x})=\alpha,x_t=0,x_{t+1},\cdots,x_n\},\]
and
\[b_{\alpha} =\# \{\bm{x}:f_g(\bm{x})=\alpha,x_t=1,x_{t+1},\cdots,x_n\},\]
where $0 \leq \alpha \leq 2^m-1$. Thus, equation (\ref{eq:wi}) is equivalent to
\begin{equation*}
  \sum_{\alpha=0}^{2^m-1}a_{\alpha}\omega_i^{\alpha}-\sum_{\alpha=0}^{2^m-1}b_{\alpha}\omega_i^{\alpha}=0 \Longleftrightarrow \sum_{\alpha=0}^{2^m-1}(a_{\alpha}-b_{\alpha})\omega_i^{\alpha}=0
\end{equation*}
Since $\omega_i$ is a $2^ith$ primitive root of unity in the complex field, and $\Phi_{2^i}(z)$ is irreducible in the integer ring, then every $\Phi_{2^i}(z)$ divide $\sum_{\alpha=0}^{2^m-1}(a_{\alpha}-b_{\alpha})z^{\alpha}$. In addition, $\Phi_{2^i}(z)$ are pairwise coprime for $1 \leq i \leq m$. Let $h(z)$ denote the product of all $\Phi_{2^i}(z)$, $1 \leq i \leq m$,
$$h(z)=(z+1)(z^2+1)\cdots(z^{2^{m-1}}+1)=\frac{z^{2^m}-1}{z-1}=1+z+z^2+\cdots+z^{2^m-1}.$$
Note that $\sum_{\alpha=0}^{2^m-1}(a_{\alpha}-b_{\alpha})z^{\alpha}$ must be a multiple of $h(z)$, and $\mbox{deg}(\sum_{\alpha=0}^{2^m-1}(a_{\alpha}-b_{\alpha})z^{\alpha})=\mbox{deg}(h(z))$, we obtain that $$a_1-b_1=a_2-b_2=\cdots=a_{2^m-1}-b_{2^m-1}.$$
Since \[\sum_{\alpha=0}^{2^m-1}a_\alpha=\sum_{\alpha=0}^{2^m-1}b_{\alpha}=2^{n-t}\Rightarrow \sum_{\alpha=0}^{2^m-1}(a_{\alpha}-b_{\alpha})=0,\]
we obtain that $a_{\alpha}-b_{\alpha}=0$.
In other words,
$$P_r \left(f_g(\bm{x})=\alpha|x_t=0,x_1,\cdots ,x_{t-1}\right)=P_r \left(f_g(\bm{x})=\alpha|x_t=1,x_1,\cdots ,x_{t-1} \right)$$
for $\forall t$ and $\forall \alpha$, i.e.,
$$P_r \left(f_g(\bm{x})=\alpha|x_1,\cdots ,x_{t-1},x_t \right)=P_r \left(f_g(\bm{x})=\alpha|x_1,\cdots ,x_{t-1}\right).$$
%
For $1\leq s\leq t-1$, let $\pi = (s, t)$ denote a transposition. Such a permutation exchange the place of two elements $s$ and $t$, leaving the others fixed. $\Phi_{2^t}(z)|(\pi \cdot F^{\{i\}})(z)$ for any $\pi=(s,t)$ is equivalent to the fact that $P_r \left( f_g(\bm{x})=\alpha\right)$ does not depend on the values of $x_1, x_2, \cdots x_t$, i.e,
\[P_r \left( f_g(\bm{x})=\alpha|x_1, x_2 \cdots x_t\right) =P_r \left(f_g(\bm{x})=\alpha\right).\]
Then considering any permutation $\pi\in S_n$, we obtain
$$P_r \left( f_g(\bm{x})=\alpha|x_{\pi(1)}, x_{\pi(2)}\cdots,x_{\pi(t)} \right)=P_r \left( f_g(\bm{x})=\alpha \right),$$
which is exactly the definition of the $(n, m, t)$-CI function.
%
\qed
\end{proof}

\begin{definition}
A generalized Boolean function $f_g$ is called a {\em symmetric function} if permuting its variables $(x_1,x_2,\cdots ,x_n)$ leads to itself.
\end{definition}
For symmetric function $f_g$, since $f_g=\pi \cdot f_g$ for any permutation $\pi \in S_n$, the second characterization for the symmetric functions is much simpler. Only $m$ points of Fourier spectra should be calculated.
\begin{corollary}\label{c1}
Let $f_g(x_1,x_2,\cdots ,x_n)$ be a symmetric generalized Boolean function. Then $f_g(x_1,x_2,\cdots ,x_n)$ is an $(n, m, t)$-CI function if and only if  \[\mathcal{F}^{\{i\}}_{f_g}(2^{n-t})=0, \] for $1 \leq i\leq m$.
\end{corollary}

%

\subsection{The Third Characterization}\label{section:s6}

The Fourier spectral characterization in section \ref{section:s5} is to regard the multi-output Boolean function as a generalized Boolean function. In this section, we give another Fourier spectral characterization for multi-output CI Boolean functions by the Fourier transform of component functions.

In paper~\cite{WangDiscrete2018}, Wang and Gong investigated the Fourier spectral characterizations of CI Boolean functions. Theorem 4 in \cite{WangDiscrete2018} showed that a Boolean function is $t$ th-order CI if and only if its Fourier spectrum vanishes at a special point for any permutation $\pi$.

\begin{fact}
\cite{WangDiscrete2018} A Boolean function $f_b$ is an $(n,1,t)$-CI function 
if and only if $\mathcal{F}_{\pi \cdot f_b}(2^{n-t})=0$ for $\forall \pi \in S_n$.
\end{fact}

It is known from Walsh spectral characterization (Fact \ref{fact1}) that a multi-output Boolean function is $t$ th-order CI if and only if all its nonzero linear combinations of the component functions of $f(\bm{x})$ are $t$ th-order CI. Then another Fourier spectral characterization is given below.
\begin{theorem}
Let $f(\bm{x})=(f_1(\bm{x}),f_2(\bm{x}),\cdots, f_m(\bm{x}))$ be a multi-output Boolean function from $\mathbb{F}^n_2$ to $\mathbb{F}_2^m$. Then $f(\bm{x})$ 
is an $(n, m, t)$-CI function if and only if
\[\mathcal{F}_{\pi \cdot (\bm{v} \cdot f(\bm{x}))}(2^{n-t}) =0,\bm{v}\neq \bm{0},\]
for $\forall \pi \in S_n$.
\end{theorem}
\begin{corollary}
Let $f(x_1,x_2,\cdots ,x_n)$ be a symmetric multi-output Boolean function from $\mathbb{F}^n_2$ to $\mathbb{F}_2^m$. Then $f(x_1,x_2,\cdots ,x_n)$ is an $(n, m, t)$-CI function if and only if
\[\mathcal{F}_{\bm{v} \cdot f(\bm{x})}(2^{n-t}) =0, \bm{v}\neq \bm{0}.\]
\end{corollary}
\section{Conclusions}\label{section:s7}
In this paper, we have studied three new characterizations for multi-output CI Boolean functions. The first characterization was given in terms of the Walsh transforms of corresponding generalized Boolean functions. The last two characterizations were obtained in terms of the Fourier transforms over the complex field. 
\begin{enumerate}
  \item A generalized Boolean function $f_g$ is an $(n, m, t)$-CI function if and only if
        $$ \sum_{\bm{x}\in \mathbb{F}_2^n}\omega_i^{f_g(\bm{x})}(-1)^{\bm{c}\cdot\bm{x}}=0,$$ for $1\leq wt(\bm{c})\leq t$ and $1 \leq i \leq m$, where $\omega_i$ is a $2^ith$ primitive root of unity in the complex field. This characterization reduces the computational complexity compared to the previous Walsh spectral characterization.
  \item A generalized Boolean function $f_g$ is an $(n, m, t)$-CI function if and only if
        $$\mathcal{F}^{\{i\}}_{\pi \cdot f_g}(2^{n-t})=0,$$ for $\forall \pi \in S_n$ and $1 \leq i \leq m$. Moreover, a symmetric generalized Boolean function $f_g$ is an $(n, m, t)$-CI function if and only if $$\mathcal{F}^{\{i\}}_{f_g}(2^{n-t})=0,$$ for $1 \leq i \leq m$.
  \item A multi-output Boolean function $f(\bm{x})$ is an $(n, m, t)$-CI function if and only if \[\mathcal{F}_{\pi \cdot (\bm{v} \cdot f(\bm{x}))}(2^{n-t}) =0, \bm{v}\neq \bm{0},\]  for $\forall \pi \in S_n$.
       A symmetric $(n, m)$-function is an $(n, m, t)$-CI function if and only if \[\mathcal{F}_{\bm{v} \cdot f(\bm{x})}(2^{n-t}) =0, \bm{v}\neq \bm{0}.\]
\end{enumerate}

The Golomb-Xiao-Massey characterization \cite{Golomb1959On,Golomb1967Shift,Xiao1988A,Golomb1999On} and the Fourier spectral characterization \cite{WangDiscrete2018} of (single-output) Boolean functions can be regarded as a special case of the results in this paper when $m=1$.
\bibliographystyle{spmpsci}      
\bibliography{references}   


\end{document}